\documentclass[11pt]{llncs}
\usepackage{amsmath}
\usepackage{framed}
\usepackage{amssymb}
\usepackage{color}
\usepackage{tikz}

\ifx\suppress\undefined
\newcommand{\TODO}[1]{%
\typeout{WARNING!!! there is still a TODO left}
\marginpar{\textbf{!TODO: }\emph{#1}}
}
\else
\newcommand{\TODO}[1]{}
\fi

\ifx\suppress\undefined
\newenvironment{todo}[1]{\noindent\rule{.3\textwidth}{1pt}\TODO{#1}\\}{\\\rule{.3\textwidth}{1pt}}
\else

\fi

\ifx\suppress\undefined
\newcommand{\NOTE}[1]{%
\typeout{WARNING!!! there are still DRAFT NOTES left}
\marginpar{!DRAFT}\emph{\textbf{DRAFT NOTES:} #1}
}
\else
\newcommand{\NOTE}[1]{}
\fi

\title{New Protocols and Lower Bound for Quantum~Secret~Sharing~with~Graph States}

\author{J\'er\^ome Javelle\inst{2}\and Mehdi Mhalla\inst{1,2} \and Simon Perdrix\inst{1,2}}

\date{}

\institute{CNRS \and LIG, Grenoble University, France}

\newcommand{\ket}[1]{\ensuremath{\left|#1\right\rangle}} 
\newcommand{\bra}[1]{\ensuremath{\left\langle#1\right|}} 
\newcommand{\rank}{\mathsf{rank}}
\begin{document}

\maketitle

\begin{abstract}
We introduce a new family of quantum secret sharing protocols with limited quantum resources which extends the protocols proposed by Markham and Sanders \cite{MS08} and Broadbent, Chouha, and Tapp \cite{Tapp}. 
Parametrized by a graph $G$ and a subset of its vertices $A$, the protocol consists in: ($i$) encoding the quantum secret into the corresponding  graph state by acting on the qubits in $A$; ($ii$) use a  classical encoding to ensure the existence of a threshold. 
These new protocols realize $((k,n))$ quantum secret sharing i.e., any set of at least $k$ players among $n$ can reconstruct the quantum secret, whereas any set of less than $k$ players has no information about the secret. In the particular case where the secret is encoded on all the qubits,  
we explore the values of $k$ for which there exists a graph such that the corresponding protocol realizes a $((k,n))$ secret sharing.  
 We show that for any threshold $k\ge n-n^{0.68}$ there exists a graph allowing a $((k,n))$ protocol. On the other hand, we prove that for any $k< \frac{79}{156}n$ there is no graph $G$ allowing a $((k,n))$ protocol. 
   As a consequence there exists $n_0$ such that the protocols introduced by Markham and Sanders in \cite{MS08} admit no threshold $k$ when the secret is encoded on all the qubits and $n>n_0$.

\end{abstract}

\begin{keywords}
Quantum Cryptography, Secret Sharing, Graphs, Graph States.
\end{keywords}

\section{Introduction}

Secret sharing schemes were independently introduced by Shamir \cite{Shamir} and Blakley \cite{B} and extended to the quantum case by Hillery  \cite{Hill} and Gottesman \cite{CG99,G00}. 

A $((k,n))$ quantum secret sharing \cite{Hill,CG99,G00} is a protocol by which a dealer distributes shares of a quantum secret to $n$ players such that any subset of at least $k$ players can reconstruct the secret by combining their shares, while any set of less than $k$ players cannot have any information about the secret. It is assumed that the secret is an arbitrary one-qubit state, that the dealer has only one copy of the secret he wants to share and that the players  can communicate together using classical and quantum channels.

A direct consequence of the no-cloning theorem \cite{nocloning} is that no $((k,n))$ quantum secret sharing protocol can exists when $k\le \frac n2$ -- otherwise two distinct sets of players can reconstruct the secret implying a cloning of quantum secret. On the other hand, for any $k>\frac n 2$ a $((k,n))$ protocol has been introduced in \cite{G00} 
 in such a way that the dimension of each share is proportional to the number of players.

The unbounded size of the share is a strong limitation of the protocol, as a consequence several schemes of quantum secret sharing  using a bounded amount of resources for each player have been introduced \cite{MS08,Tapp,qdit}. In particular, in \cite{MS08} a quantum secret sharing  scheme using graph states is presented where each player receives a single qubit. At the forefront in terms of implementation,  the graph states has emerged as a powerful and elegant family of entangled state \cite{HEB04,RB01}.

Only few threshold quantum secret sharing  schemes have been proved in the literature to be achievable using graph states: $((3,5))$ can be done using a $C_5$ graph (cycle with 5 vertices), and for any $n$, an $((n,n))$ protocol using the complete graph can be done, up to some constraints on the quantum secret  \cite{MS08}. Independently \cite{Tapp} introduced an $((n,n))$ protocol for any $n$. This protocol is based on the GHZ state \cite{GHZ} which is locally equivalent to a complete graph state \cite{HEB04}. The technique which consists in mixing the quantum secret before to encode it into a larger state is also used in \cite{NMH} in such a way that some players have a classical share but no quantum share.

  We introduce a new family of secret sharing protocols using graph states.  Like in \cite{MS08} the quantum secret is encoded into a graph state shared between the players, but in order to obtain threshold protocols, an additional round is added to the protocol. This round consists in mixing the quantum secret using a one-time pad  scheme which classical key is then shared between the players using a classical secret sharing protocol. This technique extends the one presented in \cite{Tapp} in which the secret is partially mixed and then shared using a fixed quantum state, namely the GHZ state which is equivalent to the complete graph state.    Independently, a hybrid classical-quantum construction of QSS has been rencently proposed in \cite{FG11} where they optimize the quantum communication complexity when the size of the secret is greater  than the number of players. 
  
  The family of protocols we introduce in the present paper is parametrised by a pair $(G,A)$ where $G$ is a graph and $A$ is a non empty set of vertices of the graph. We explore the possible values of $k$ for which there exists a pair $(G,A)$ leading to a $((k,n))$ protocol. Our main results are: first, we introduce a family of graphs which can realise any $((k,n))$  protocol when $k>n-n^{0.68}$. This result proves that graph states secret sharing can be used not only for $((n,n))$ protocols, but also for any threshold larger than $n-n^{0.68}$. The second main result of the paper is the proof that there is no graph $G$ such that $(G,V(G))$  realizes a $((k,n))$ protocols when $k<  \frac {79} {156} n$. Notice that this lower bound also applies in the protocol introduced by Markham and Sanders.  Moreover, it suggests that secret sharing protocols with a threshold closed to half of the players cannot  be achieve with shares of bounded size.

Section \ref{sec:gs} is dedicated to the description of the various secret sharing protocols based on graph states: section \ref{sec:cqss} describes the protocol cQSS introduced in \cite{MS08} for sharing a classical secret, while section \ref{sec:qqss} describes its extension to a quantum secret qQSS. We prove that the sufficient combinatorial conditions for accessibility and privacy introduced in \cite{KMMP} for these protocols are actually necessary. These graphical characterisation are key ingredients for proving the main results of this paper. In section \ref{sec:qqsss} the new family of protocols qQSS* is introduced. In section \ref{sec:c5} we prove that for any $k>n-n^{0.68}$ there exists a qQSS* protocol which realises a $((k,n))$ secret sharing. Finally, in section \ref{sec:lowerbound}, we prove the following lower bound: for any $k<\frac {79}{156}n$ there is no graph $G$ such that $(G,V(G))$  realises a $((k,n))$ qQSS* scheme. A preliminary version of this work has been presented at \cite{JMP11}.

\section{Graph state secret sharing}\label{sec:gs}

\subsection{Sharing a classical secret using a graph state}\label{sec:cqss}

For a given graph $G$ on $n$ vertices $v_1,\ldots, v_n$, the corresponding graph state $\ket G$ is a $n$-qubit quantum state defined as $$\ket G = \frac 1{\sqrt{2^n}}\sum_{x\in \{0,1\}^n}(-1)^{q(x)}\ket x$$ where $q(x)$ is the number of edges in the induced subgraph $G_x = (\{v_i\in V(G) ~|~x_i=1\}, \{(v_i,v_j) \in E(G) ~|~x_i=x_j=1\})$.

Graph states have the following fundamental fixpoint property: given a graph $G$, for any vertex $u\in V(G)$, $$X_uZ_{N(u)}\ket G =\ket G$$ where $N(u)$ is the neighborhood of $u$ in $G$, $X=\ket x\mapsto \ket {\bar x}$, $Z=\ket x \mapsto (-1)^x\ket x$ are one-qubit Pauli operators and $Z_A = \bigotimes_{u\in A}Z_u$ is a Pauli operator acting on the qubits in $A$. As a consequence, for any subset $D\subseteq V(G)$ of vertices, $\bigotimes_{u\in D} X_u Z_{N(u)}\ket G = \ket G$. Since $X$ and $Z$ anti-commutes and $Z^2=X^2=I$, $$ (-1)^{|D\cap Odd(D)| }X_DZ_{Odd(D)}\ket G = \bigotimes_{u\in D} X_u Z_{N(u)} \ket G = \ket G$$ where $Odd(D) := \{v\in V(G) ~s.t.~|N(v)\cap D|=1 \bmod{2}\}$ is the odd neighborhood of $D$.

We present a family of quantum protocols for sharing a classical secret ($cQSS$) parametrized  by a graph $G$ and a non empty subset $A$ of the vertices of the graph. This family of protocols has been introduced in \cite{MS08}. Obviously, sharing a classical bit can be done using a classical scheme, like \cite{Shamir}, instead of using a quantum state. However, the study of the cQSS protocols, and in particular the characterisation of accessibility and privacy (see  corollary \ref{cor:charac}) are essential for  the next sections where the sharing of a quantum secret is considered.

Suppose a dealer wants to share a classical secret $s\in\{0,1\}$ between $n=|V(G)|$ players.  
The dealer prepares the state $\ket {G_s} = Z_A^s\ket G$ 
 where $Z_A^0$ is the identity and $Z_A^1$ consists in applying the Pauli operator $Z$ on each qubit of $A$. 
The dealer sends  each player $i$ the qubit $q_i$ of $\ket {G_{s}}$. 
Regarding the reconstruction of the secret, a set $B$ of players can recover the secret if and only if $tr(\rho_B(0)\rho_B(1))=0$, i.e. if the set of players can distinguish perfectly between the two states $\rho_B(0)$ and $\rho_B(1)$, where $\rho_B(s) = tr_{V\setminus B} (\ket{G_s}\bra{G_s})$ is  the state of the subsystem of the players in $B$. On the other hand, a set $B$ of players has no information about the secret if and only if $\rho(0)$ and $\rho(1)$ are indistinguishable, i.e. $\rho(0)=\rho(1)$.

Sufficient graphical conditions for accessibility and privacy have been proved in \cite{KMMP}:

\begin{lemma}[\cite{KMMP}]
Given a cQSS protocol $(G,A)$, for any $B\subseteq V(G)$,
\\{-- If $ \exists D\subseteq B$ s.t. $D\cup Odd(D)\subseteq B$ and $|D\cap A|=1 \bmod{2} $ then  $ B$ can access the secret.}\\ 
{-- If $\exists C\subseteq \overline B = V(G) \setminus B$ s.t. $Odd(C)\cap B = A\cap B $ then  $B$ has no information about the secret.}
\end{lemma}

According to the previous lemma, for a given set of players $B\subseteq V(G)$, if $\exists  D\subseteq B$ s.t. $D\cup Odd(D)\subseteq B$ and $|D\cap A|=1 \bmod{2} $ then  $ B$ can access the secret. More precisely, the players in $B$ perform a measurement of their qubits according to the observable $(-1)^{|D\cap Odd(D)|}X_DZ_{Odd(D)}$. This measurement produces a classical outcomes $s\in \{0,1\}$ which is the reconstructed secret \cite{KMMP}.

We prove that the sufficient graphical conditions are actually necessary conditions, and that any set of players is either able to access the secret or has no information about the secret.

\begin{lemma}
Given a graph $G$ and $A\subseteq V(B)$, for any $B\subseteq V(G)$, $B$ satisfies exactly one of the two properties:\\
$~~~i.$$~~\exists D\subseteq B, D\cup Odd(D)\subseteq B \wedge |D\cap A|=1 \bmod{2}$\\
$~~~ii.$ $\exists C\subseteq V\setminus B, Odd(C)\cap B = A\cap B$%
\end{lemma}

\begin{proof}
For a given $B\subseteq V(G)$, let $\Gamma_{B, \overline B}$ be the cut matrix induced by $B$, i.e. the sub-matrix of the adjacency matrix $\Gamma$ of $G$ such that  the columns of $\Gamma_{B,\overline B}$ correspond to the vertices in $B$ and its rows to the vertices in $\overline B$. $\Gamma_{B, \overline B}$ is the matrix representation of the linear function which maps every $X\subseteq B$ to $\Gamma_{B,\overline B}.X = Odd(X)\cap \overline  B$, where the set $X$ is identified with its characteristic column vector.  Similarly, $\forall Y\subseteq \overline B$, $\Gamma_{\overline B,B}.Y = Odd(Y)\cap B$ where $\Gamma_{\overline B, B} = \Gamma_{B,\overline B}^T$ since $\Gamma$ is symmetric. Moreover, notice that for any set $X,Y\subseteq V(G)$, $|X\cap Y| \bmod{2}$ is given by the matrix product $Y^T.X$ where again sets are identified with their column vector representation. Equation $(i)$ is satisfied iff $\exists D$ s.t. $\left(\frac{(A\cap B)^T}{\Gamma_{B,\overline B}}\right).D = \left(\frac 1 0\right)$ which is equivalent  to $\rank\left(\frac{(A\cap B)^T}{\Gamma_{B,\overline B}}\right)  = \rank\left(\frac{(A\cap B)^T|~1}{~\Gamma_{B,\overline B}~~|~0}\right) = \rank\left(\frac{~~~0~~|~1}{\Gamma_{B,\overline B}~|~0}\right) =\rank(\Gamma_{B,\overline B})+1$. Thus $(i)$ is true iff $R(B) = 1$ where $R(B):=\rank\left(\frac{(A\cap B)^T}{\Gamma_{B,\overline B}}\right) - \rank(\Gamma_{B,\overline B})$. Similarly equation $(ii)$ is satisfied iff $\exists C$ s.t. $\Gamma_{\overline B,B}.C = A\cap B$ iff $\rank(\Gamma_{\overline B,B}|A\cap B) = \rank(\Gamma_{\overline B,B})$. Thus $(ii)$ is true iff  $R(B)=0$. Since for any $B\subseteq V(G)$,  $R(B)\in \{0,1\}$ it comes that either $(i)$ is true or $(ii)$ is true. \hfill$\Box$ \end{proof}

\begin{corollary}\label{cor:charac}
Given a cQSS protocol $(G,A)$, for any $B\subseteq V(G)$,\\
\centerline{ $ \exists D\subseteq B, D\cup Odd(D)\subseteq B ~\wedge~ |D\cap A|=1 \bmod{2}  \iff  B$ can access the secret.}
\centerline{ $\exists C\subseteq V\setminus B, Odd(C)\cap B = A\cap B \iff B$ has no information about the secret. }

\end{corollary}

\subsection{Sharing a quantum secret}\label{sec:qqss}
 
Following \cite{MS08},  the cQSS protocols are extended to qQSS schemes for sharing a quantum secret $\ket \phi = \alpha \ket 0+\beta \ket 1$. Given a graph $G$ and $A$ a non empty subset of vertices, the dealer prepares the quantum state $\ket {G_\phi} = \alpha \ket {G_0}+ \beta\ket{G_1}$. Notice that the transformation $\ket \phi \mapsto \ket {G_\phi}$ is a valid quantum evolution -- i.e. an isometry -- whenever $\ket {G_0}$ is orthogonal to $\ket {G_1}$ which is guaranteed by $A\neq \emptyset$.  Then, the dealer sends  each player $i$ the qubit $q_i$ of $\ket{G_\phi}$. Regarding the reconstruction of the secret, it has been proved in \cite{MS08}, that a set $B$ of players can recover the quantum state $\ket \phi$ if and only if $B$ can reconstruct a classical secret in the two protocols $cQSS(G,A)$ and $cQSS(G\Delta A,A)$, where $G\Delta A =(V(G), E(G)\Delta (A\times A))$ and $X\Delta Y = (X\cup Y)\setminus (X\cap Y)$ is the symmetric difference. In other words $G\Delta A$ is obtained by complementing the egdes of $G$ incident to two vertices in $A$. We introduce an alternative characterisation of q-accessibility (ability to reconstruct a quantum secret) which does not involved the complemented graph $G\Delta A$:

 \begin{lemma}
Given a qQSS protocol $(G,A)$, a set $B$ of players can reconstruct the quantum secret if and only if, in the protocol cQSS $(G,A)$, $B$ can reconstruct the classical secret and $\overline B = V(G)\setminus B$ cannot.  
 \end{lemma}

 \begin{proof}
 First notice that for any $X$, if $|X\cap A| = 1 \bmod{2}$ then $Odd_{G\Delta A}(X) = Odd_G(X)\Delta A$. Thus for any $X,Y$, if $|X\cap A| = 1 \bmod{2}$, $Odd_{G\Delta A}(X)\cap Y=\emptyset \iff (Odd_G(X)\Delta A)\cap Y = \emptyset \iff (Odd_G(X)\cap Y)\Delta (A\cap Y)=\emptyset \iff Odd_G(X)\cap Y = A\cap Y$. \\
 $(\Rightarrow)$ Assume that $B$ can reconstruct the quantum secret, so $B$ can reconstruct the classical secret in $G\Delta A$. Thus $\exists D\subseteq B$ s.t. $Odd_{G\Delta A}(D)\cap \overline B=\emptyset$. According to the previous remark, it implies that $Odd_G(D)\cap \overline B=A\cap \overline B$, so $\overline B$ cannot reconstruct the secret.\\
 $(\Leftarrow)$ Assume $\overline B$ cannot recover the classical secret and $B$ can. So $\exists C\subseteq B$ s.t. $Odd_G(C) \cap B=A\cap B$. If $|C\cap A| $ is even, let $C':= C\Delta D$ where $|D\cap A|$ is odd and $Odd_G(D)\cap B = \emptyset$. Such a set $D$ exists since $B$ can reconstruct the classical secret in $G$. If $|C\cap A| $ is odd, then let $C':=C$. In both cases, $|C'\cap A| = 1 \bmod{2}$ and $Odd_G(C')\cap B = A\cap B$, so according to the previous remark, $Odd_{G\Delta A}(C')\cap B= \emptyset$, as a consequence $B$ can access the classical secret in $G\Delta A$.  \hfill $\Box$
 \end{proof}

In any pure quantum secret sharing protocol a set of players can reconstruct a quantum secret if and only if its complement set of players has no information about the secret (see \cite{G00}). As a consequence:

\begin{corollary}
Given a qQSS protocol $(G,A)$, a set $B$ of players has no information about the quantum secret if and only if, in the protocol cQSS $(G,A)$, $\overline B$ can reconstruct the classical secret and $B$ cannot. 
\end{corollary}

Sets of players that can reconstruct the secret and those who have no information about the secret admit simple graphical characterisation thanks to the simple reduction to the classical case. However, contrary to the cQSS case, there is a third kind of set players, those who can have some information about the secret but not enough to reconstruct the secret perfectly. For instance for any $n>1$ consider the qQSS protocol $(K_n,\{v_1,\ldots ,v_n\})$ where $K_n$ is the complete graph on the $n$ vertices $v_1,\ldots v_n$.  For any set $B$ of vertices s.t. $B\neq \emptyset$ and $\overline B\neq \emptyset$,  both $B$ and $\overline B$ cannot reconstruct a classical secret in the corresponding $cQSS$ protocol, so $B$ cannot reconstruct the quantum secret perfectly but has some information about the secret. 

\begin{corollary}
\label{autocomp}
 The  qQSS protocols $(G,A)$ and $(G\Delta A,A)$ have the same accessing structures -- i.e. a set of players can access the secret in $(G,A)$ iff it can access the secret in $(G\Delta A, A)$. In particular, the protocols $(G,V(G))$ and $(\overline G, V(G))$ have the same accessing sets.  
\end{corollary}

\subsection{Threshold schemes}\label{sec:qqsss}

For any qQSS protocol $(G,A)$, the accessing structures can be characterized. For secret sharing protocols, this is often interesting to focus on threshold protocols, i.e. protocols such that there exists an integer $k$ such that any set of least $k$ players can reconstruct the secret, whereas any set of at most $k-1$ players have no information about the secret. Such threshold protocols are denoted $((k,n))$. In \cite{G00}, it has been proved that if the dealer is sending a pure quantum state to the players, like in the qQSS protocols, then the threshold, if it exists, should be equal to $\frac{n+1}2$ where $n$ is the number of players. This property which is derived from the no-cloning theorem, is very restrictive. It turns out that there is a unique threshold for which a qQSS protocol is known. This protocol is  $(C_5,\{v_1,\ldots v_5\})$ where $C_5$ is the cycle graph on $5$ vertices. The threshold for this protocol is $3$. In section \ref{sec:lowerbound}, we  prove that under the constraint that $A = V(G)$, there is no threshold qQSS protocol for $n>79$. 

However, in general a qQSS protocol corresponds to   a ramp secret sharing scheme  \cite{Oga05} where any set of players smaller than $n-k$ cannot access the information and any set greater than $k$ can. In this section we show how these ramp schemes can be turned into threshold schemes by adding a classical secret sharing round.

\begin{theorem}
\label{qQSS*}
Given a graph $G$ over $n$ vertices, a non empty subset of vertices $A$, and an integer $k$, such that  $\forall B\subseteq V(G)$ with  $|B|=k$, 
 $\exists C_B,D_B\subseteq B$
satisfying $|D_B\cap A|=1 \bmod 2 $,  $Odd(D_B)\subseteq B$, and  $Odd(C_B)\cap \overline B= A\cap \overline B$, it exists an $((k+c,n+c))$ quantum secret sharing protocol for any $c\ge 0$ in which the dealer  sends   one qubit to $n$ players and uses a $(k+c)$-threshold classical secret sharing scheme on the $n+c$ players. 
\end{theorem}

The rest of the section is dedicated to define a family of protocols called $qQSS^*$ satisfying the theorem.

Inspired by the work of Broadbent, Chouha and Tapp \cite{Tapp}, we extend the qQSS scheme adding a classical reconstruction part. In \cite {Tapp}, 
a family of unanimity -- i.e. the threshold is the number of players -- quantum secret sharing protocols have been introduced. They use a $GHZ$ state which equivalent to the graph state $\ket {K_n}$ where $K_n$ is the complete graph on $n$ vertices. We extend this construction to any graph, using also a more  general initial encryption corresponding to a quantum one-time pad of the quantum secret.

\noindent {\bf Quantum secret sharing with graph states and classical reconstruction ($qQSS^*$).}
Given a graph $G$, a non empty subset of players $A$, and an integer $k$, such that $\forall B\subseteq V(G)$ if $|B|\ge k$ then 
 $\exists C_B,D_B\subseteq B$ 
satisfying $|D_B\cap A|=1 \bmod 2 $,  $Odd(D_B)\subseteq B$, and  $Odd(C_B)\cap \overline B= A\cap \overline B$. \\ Suppose the dealer wishes to share the quantum secret $\ket \phi = \alpha \ket 0+\ket 1$. 
\begin{itemize}
\item {\bf Encryption.} The dealer chooses uniformly at random $b_x,b_z\in \{0,1\}$. and apply $X^{b_x}Z^{b_z}$ on $\ket \phi$.
The resulting state is $\ket {\phi'} = \alpha \ket{b_x}+\beta (-1)^{b_z}\ket {\overline {b_x}}$. 
\item {\bf Graph state embedding.} The dealer embeds $\ket {\phi'}$ to the $n$-qubit state $\alpha \ket {G_{b_x}}+\beta (-1)^{b_z} |G_{\overline{b_x}}\rangle$. 
\item {\bf Distribution.} The dealer sends each player $i$ the qubit $q_i$. Moreover using a classical secret sharing scheme with a threshold $k$, the dealer shares the bits $b_x, b_z$. 
\item {\bf Reconstruction.} The reconstruction of the secret  for a set $B$ of players s.t. $|B|\ge k$ is in $3$ steps: first the set $D_B$ is used to add an ancillary qubit and put the overall system in an appropriate state; then $C_B$ is used to disentangled the ancillary qubit form the rest of the system; finally the classical bits $b_x$ and $b_z$ are used to recover the secret:
\subitem{(a)} 
The players in $B$ applies on their qubits  the isometry $U_{D_B}:=\ket 0 \otimes P_0 + \ket 1\otimes P_1$ where $P_i$ are the projectors associated with observable $\mathcal O_{D_B} = (-1)^{|D_B\cap Odd(D_B)|}X_{D_B}Z_{Odd({D_B})}$, i.e. $P_i := \frac{I+(-1)^i \mathcal O_{D_B}}2$. The resulting state is $\alpha\ket {b_x}\otimes \ket{G_{b_x}} + \beta.(-1)^{b_z}\ket {\overline {b_x}}\otimes | {G_{\overline{b_x}}}\rangle$.
\subitem{(b)} 
The players in $B$ apply the controlled unitary map $\Lambda_{V_{C_B}} = \ket 0\bra 0\otimes I + \ket 1\bra 1\otimes V_{C_B}$, 
where  $V_C:=(-1)^{|C\cap Odd(C)|}X_C Z_{Odd(C)\Delta A}$. 
The resulting state is $\alpha \ket {b_x}\otimes \ket {G} + \beta .(-1)^{b_z}\ket {\overline {b_x}}\otimes \ket {G} = \left(\alpha \ket {b_x}+\beta.(-1)^{b_z}\ket {\overline {b_x}}\right)\otimes \ket {G}$. 
\subitem{(c)} Thanks to the classical secret sharing scheme, the players in $B$ recover the bits $b_x$ and $b_z$. They apply $X^{b_x}$ and then $Z^{b_z}$ for reconstructing the quantum secret $\alpha \ket 0 +\beta \ket 1$ on the ancillary qubit. 
\end{itemize}

Note that this reconstruction method can be used for the qQSS protocols defined in \cite{KMMP} and for which the reconstruction part was not explicitly defined.

\begin{lemma}
A $qQSS^*$ protocol $(G,A,k)$ is a $((k,|V(G)|))$ secret sharing protocol if for any $B\subseteq V(G)$ s.t. $|B|\ge k$, $B$ can reconstruct the secret in $qQSS (G,A)$.
\end{lemma}

\begin{proof}
The classical encoding ensures that any set of size smaller then $k$ cannot access to the secret. 
$\mathcal O_{D_B}$ is acting on the qubits $D_B\cap Odd(D_B)\subseteq B$. Moreover $P_i\ket {G_s} = \ket {G_s}$ if $i=s$ and $0$ otherwise, so the application of the isometry $U_{D_B}$ produce the state $\alpha\ket {b_x}\otimes \ket{G_{b_x}} + \beta.(-1)^{b_z}\ket {\overline {b_x}}\otimes | {G_{\overline{b_x}}}\rangle$. Regarding step $b$ of the reconstruction, 
since $Odd(C) \cap \overline B = A\cap \overline B$, 
$C\cup Odd(C)\Delta A \subseteq B$
 $V_C$ is acting on the qubits 
 in $B$. 
Moreover $V_C$ produces the states  $ \left(\alpha \ket {b_x}+\beta.(-1)^{b_z}\ket {\overline {b_x}}\right)\otimes \ket {G}$. Finally the classical secret scheme guarantees that the players in $B$ have access to $b_x$ and $b_z$ so that they reconstruct the secret. 
\hfill $\Box$
\end{proof}

\noindent{\bf Proof of Theorem \ref{qQSS*}.} 
The correctness of the qQSS* protocol implies   that given a graph $G$ over $n$ vertices, a non empty subset of vertices $A$, and an integer $k$, such that  $\forall B\subseteq V(G)$ with  $|B|=k$, $\exists C_B,D_B\subseteq B$
satisfying $|D_B\cap A|=1 \bmod 2 $,  $Odd(D_B)\subseteq B$, and  $Odd(C_B)\cap \overline B= A\cap \overline B$, it exists a $((k,n)$ protocol. In order to finish the proof of Theorem \ref{qQSS*} this protocol is turned into a $((k+c,n+c))$ protocol for any $c\ge 0$  the qQSS* protocol is modified as follows, following the technique used in \cite{NMH}. During the distribution stage, the dealer shares $b_x$ and $b_z$ with all the $n+c$ players with a threshold $k+c$, but sends a qubit of the graph state to only $n$ players chosen at random among the $n+c$ players. During the reconstruction, a set of $k+c$ players must contain at least $k$ players having a qubit. This set use the reconstruction steps $a$ and $b$ and then the last step $c$ is done by all the $k+c$ players. $\hfill \Box$

In the following we focus on the particular case where $A=V(G)$.

\section{Building $((n-n^{0.68},n))$-$qQSS^*$ Protocols }\label{sec:c5}

We give a construction of an infinite family of quantum secret sharing schemes $\left( \left( k,n \right) \right)$ where $k=n-n^{\frac{log(3)}{log(5)}} < n-n^{0.68}$. 
This construction can be defined recursively from cycle over 5 vertices ($C_5$)  which has been used in Markham and Sanders  \cite{MS08} to build a ((3,5)) quantum secret sharing protocol. 

We define a composition law $*$ between two graphs $G_1 = (V_1, E_1)$ and $G_2 = (V_2, E_2)$ as follows:

\begin{definition}
	Let $G_1 = (V_1, E_1)$ and $G_2 = (V_2, E_2)$ be two graphs.
	The lexicographic product $*$ of two graphs $G_1$ and $G_2$ is defined as follows:
	$V(G_1*G_2) := V_1\times V_2$ and $E(G_1*G_2):=\{(( u_1, u_2), ( v_1, v_2))~|~ (u_1,v_1)\in E_1 \text {~or~} ( u_1 =v_1 \wedge (u_2,v_2)\in E_2)\}$.

\end{definition}

In other terms, the graph $G$ is a graph $G_1$ which vertices are replaced by copies of the graph $G_2$, and which edges are replaced by complete bipartitions between two copies of the graph $G_2$ (Figure \ref{Composition}).

\begin{figure}[h!]
		\[\vcenter{\begin{tikzpicture}[scale=0.4]
		\tikzstyle{vertex}=[circle,fill=black!25,minimum size=12pt,inner sep=0pt]
			\node[vertex] (v1) at (0,1){1};
			\node[vertex] (v2) at (2,1){2};
			\node[vertex] (v3) at (4,0){3};
			\node[vertex] (v4) at (5,2){4};
			\draw[thick] (v1) -- (v2) -- (v3);
			\draw[thick] (v2) -- (v4);
			\node[] (composition) at (7,1){ $*$};
			\node[vertex] (w1) at (9,0){1};
			\node[vertex] (w2) at (10,2){2};
			\node[vertex] (w3) at (11,0){3};
			\draw[thick] (w1) -- (w2) -- (w3) -- (w1);
		\end{tikzpicture}} \hspace{-8cm}\vcenter{=}\hspace{-12.5cm}\vcenter{
		\begin{tikzpicture}[scale=0.4]
		\tikzstyle{vertex}=[circle,fill=black!25,minimum size=12pt,inner sep=0pt]
			\node[vertex] (v11) at (0,2){1,1};
			\node[vertex] (v12) at (1,4){1,2};
			\node[vertex] (v13) at (2,2){1,3};
			\node[] (u1) at (2,3) {};
			\node[] (u2) at (6,3) {};
			\node[vertex] (v21) at (6,2){2,1};
			\node[vertex] (v22) at (7,4){2,2};
			\node[vertex] (v23) at (8,2){2,3};
			\node[] (u3) at (8,3) {};
			\node[] (u4) at (12,1) {};
			\node[vertex] (v31) at (12,0){3,1};
			\node[vertex] (v32) at (13,2){3,2};
			\node[vertex] (v33) at (14,0){3,3};
			\node[] (u5) at (14,5) {};
			\node[vertex] (v41) at (14,4){4,1};
			\node[vertex] (v42) at (15,6){4,2};
			\node[vertex] (v43) at (16,4){4,3};
			\draw[thick] (v11) -- (v12) -- (v13) -- (v11);
			\draw[thick] (v21) -- (v22) -- (v23) -- (v21);
			\draw[thick] (v31) -- (v32) -- (v33) -- (v31);
			\draw[thick] (v41) -- (v42) -- (v43) -- (v41);
			\draw[line width=1mm] (u1) -- (u2);
			\draw[line width=1mm] (u3) -- (u4);
			\draw[line width=1mm] (u3) -- (u5);
		\end{tikzpicture}}\]
\caption{Graphical explanation of the composition law $*$ between two graphs.
	       A thick line represents a complete bipartition between two triangle graphs.}
\label{Composition}
\end{figure}
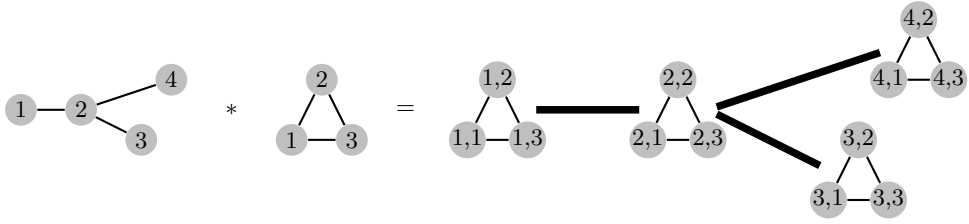

\begin{lemma}
\label{complProduct} For any graphs $G_1$ and $G_2$, $\overline{G_1 * G_2} = \overline{G_1}*\overline{G_2}$
\end{lemma}

\begin{proof}
	It is easy to see that $V(\overline{G_1 * G_2}) = V(\overline{G_1}*\overline{G_2})$.
	We want to show that $E(\overline{G_1 * G_2}) = E(\overline{G_1}*\overline{G_2})$.
	
	Consider an edge $\big( (u_1,u_2),(v_1,v_2) \big) \in E(\overline{G_1 * G_2})$.
	If $u_1 = v_1$ then $(u_2,v_2) \notin E_2$.
	Otherwise $(u_1, v_1) \notin E_1$.
	Thus
		$( u_1, v_1) \in \overline{E_1}$  or  $( u_1 = v_1 \text{ and } ( u_2, v_2) \in \overline{E_2})$
	which means, by definition, that $\big( (u_1,u_2),(v_1,v_2) \big) \in E(\overline{G_1} *\overline{ G_2})$.
	Therefore, $E(\overline{G_1* G_2}) \subseteq E(\overline{G_1} *\overline{ G_2})$.
	
	Furthemore $|E(\overline{G_1} *\overline{ G_2})|=n_1 |E(\overline{G_2})| + n_2^2 |E(\overline{G_1})|$ where
	$n_1 = |V_1|$ and $n_2 = |V_2|$. 
	Thus $|E(\overline{G_1} *\overline{ G_2})|=n_1 n_2 (n_1 n_2-1)/2 - ( n_1 E(G_2) + n_2^2 E(G_1))= |E(\overline{G_1* G_2})|$. Therefore $E(\overline{G_1 * G_2}) = E(\overline{G_1}*\overline{G_2})$.	\hfill $\Box$
\end{proof}

\begin{lemma}
\label{ProductThreshold}
	Let $G_1$, $G_2$ be two graphs such that $G_1 = (V_1,E_1)$ realizes an $(( k_1, n_1 ))$ $qQSS^*$ protocol and $G_2 = (V_2,E_2)$ realizes a $(( k_2, n_2 ))$ $qQSS^*$protocol. The graph $G = G_1*G_2 = (V,E)$ realizes an$ (( k,n )) $ $qQSS^*$ protocol where
	$$
	\left\{
	\begin{array}{l}
	n = n_1n_2 \\
	k = n_1n_2 - (n_1-k_1+1)(n_2-k_2+1) +1
	\end{array}
	\right.
	$$
\end{lemma}

\begin{proof}
	First we  show that if   $G_1 = (V_1,E_1)$ realizes an $(( k_1, n_1 ))$ $qQSS^*$ protocol and $G_2 = (V_2,E_2)$ realizes an $(( k_2, n_2 ))$ $qQSS^*$, then, in the graph $G = G_1*G_2 = (V,E)$, for   any set $B \subseteq V$ of size $k$ (with $k= n_1n_2 - (n_1-k_1+1)(n_2-k_2+1) +1$ it exists a set $D_B$ such that $|D_B|=1 \bmod 2 $,  $Odd(D_B)\subseteq B$.
	

	For any set $B \subseteq V$ and any vertex $v_1 \in V$, let $B_2(v_1) = \{ v_2 \in V_2 ~s.t.~ (v_1, v_2) \in B \}$ and $B_1= \{ v_1 \in V_1 ~s.t.~ |B_2(v_1)| \geq k_2 \}$.
	
	We claim that for all set $B \subseteq V$ of size $|B| = k$, the size of the set $B_1$ verifies $|B_1| \geq k_1$.
	
	By contradiction, notice that $B=\bigcup_{v_2 \in B_2(v_1), v_1 \in V_1} \{ (v_1,v_2)\}$. Therefore:
	$|B| = |V|- \sum_{v_1 \in B_1}  |V_2 \setminus B_2(v_1)| - \sum_{v_1 \in {V_1 \setminus B_1}}  |V_2 \setminus B_2(v_1)|$.
	Thus $|B| \le  n_1n_2 - |V_1 \setminus B_1| (n_2-k_2+1) \le k-1$  if $|B_1|\le k_1$.

	Now we consider any set $B \subseteq V$ of size $|B| = k$. As $|B_1| \geq k_1$,  it exists a set $ D_1 \subseteq B_1$ with  $|D_1| = 1 \bmod{2}$ and $D_1 \cup Odd(D_1) \subseteq B_1$. 
	
	Furthermore for any $v_1\in B_1$, $|B_2(v_1)|\ge k_2$   and thus there exists   $D_2(v_1) \subseteq B_2(v_1))$ with $|D_2(v_1)| = 1 \bmod{2}$ and $D_2(v_1) \cup Odd(D_2(v_1)) \subseteq B_2(v_1)$ and there exist $C_2(v_1) \subseteq B_2(v_1)$ with  $ V_2 \setminus B_2(v_1) \subseteq Odd(C_2(v_1)))$.
	
	Let $C_2^0(v_1)= C_2(v_1)$ if $| C_2(v_1)| = 0 \bmod{2}$ and $C_2^0(v_1)\Delta  D_2(v_1)$ otherwise, and let  $ C_2^1(v_1)= C_2^0(v_1)\Delta D_2(v_1)$.

	We partition $V_1$ in 4 subsets and define for any vertex $v_1$ a set $S_2(v_1)\subseteq V_2$ as follows
	$$
	\left\{
	\begin{array}{lll}
	\text{If } v_1\in D_1 \cap Even(D_1) &,& S_2(v_1)=D_2(v_1)\\
	\text{If } v_1\in D_1 \cap Odd(D_1) &,& S_2(v_1)=C_2^1(v_1))\\
	\text{If } v_1\in V_1 \setminus D_1 \cap Even(D_1) &,& S_2(v_1)=\emptyset \\
	\text{If } v_1\in V_1 \setminus D_1\cap Odd(D_1) &,& S_2(v_1)=C_2^0(v_1)\\
	\end{array}
	\right.
	$$
	
		\noindent Consider the set $D_B = \bigcup_{v_1 \in V_1} \{v_1\} \times S_2(v_1)$, $D_B \subseteq B$ and $|D_B| = $ \\$ \sum_{v_1 \in D_1 \cap Even(D_1)} |D_2(v_1)| + \sum_{v_1 \in D_1 \cap Odd_(D_1)} |C_2^1(v_1)| $ $+ \sum_{v_1 \in V_1 \setminus D_1 \cap Odd(D_1)} |C_2^0(v_1)| $. 	Therefore $|D_B|= |D_1| =1\bmod{2}$.
		
	For each $v = (v_1,v_2) \in V \setminus B$, 
		$\left| \mathcal{N}_{G}( v) \cap D_B \right| $ $ = \left| \mathcal{N}_{G_2}( v_2) \cap S_2(v_1) \right| $ \\$+ \sum_{u_1 \in \mathcal{N}_{G_1}( v_1)} \left| S_2(u_1) \right| $.  If $v_1 \in V_1 \setminus D_1$, then $|S_2(v_1)|=0 \bmod{2}$, thus \\ $\left| \mathcal{N}_{G}( v) \cap D_B \right|= \left| \mathcal{N}_{G_2}( v_2) \cap S_2(v_1) \right| + \left| \mathcal{N}_{G_1}( v_1) \cap D_1 \right| \mod{2} $.

	Furthermore, if $v_1\in Even(D_1)$,  $\left| \mathcal{N}_{G_2}( v_2) \cap S_2(v_1) \right| = \left| \mathcal{N}_{G_1}( v_1) \cap D_1 \right|=0  \mod{2} $   and if $v_1\in Odd(D_1)$,  $\left| \mathcal{N}_{G_2}( v_2) \cap S_2(v_1) \right| = \left| \mathcal{N}_{G_1}( v_1) \cap D_1 \right| =1 \mod{2}$.
	  
	
	Therefore  $\left| \mathcal{N}_{G}( v) \cap D_B \right|=0 \bmod{2}$ which implies that  $D_B \cup Odd(D_B) \subseteq B$.

	 Furthermore, using  Lemma \ref{complProduct}, we have $\overline{G_1 * G_2} = \overline{G_1}*\overline{G_2}$.
	And from corollary \ref{autocomp}, as $G_1$ and $G_2$ realize $qQSS*$ protocols $\overline{G_1}$ and $\overline{G_2}$ have the same threshold .Therefore,  in $\overline{G_1*G_2}$ it exists a set $D'_B$ such that its odd neighborhood in the complementary graph satisfies $Odd_{\overline{G_1*G_2}}(D'_B)\cap V\setminus B=\emptyset$   thus $Odd_{G_1*G_2}(D'_B)\cap V\setminus B=V\setminus B$ and $D'_B$ is a valid $C_B$ to define an $((k,n)) $ $qQSS^*$ protocol.
\hfill $\Box$	
\end{proof}

Let us define $G^{*i} = \underbrace{G * G * \cdots *G}_{i \text{ times}}$.
If $G$ realizes a $((k,n))$ protocol, we want to find $n_i$ and $k_i$ such that $G^{*i}$ realizes a $((k_i,n_i))$ protocol:
\begin{lemma}
\label{G*i}
	Let $G$ be a graph which realizes a $((k,n))$ protocol.
	Then the graph $G^{*i}$ realizes a $((k_i,n_i))$ protocol where
	$$
	\left\{
	\begin{array}{l}
	n_i = n^i \\
	k_i = n^i - (n-k+1)^i +1
	\end{array}
	\right.
	$$
\end{lemma}

\begin{proof}
	By induction: first, we notice that $n_1 = n$ and $k_1=k$.
	Then, if we write $G^{*i+1} = G * G^{*i}$, from Lemma \ref{ProductThreshold} 
$n_{i+1} = n.n_i$ and $k_{i+1} = n.n_i - (n-k+1)(n_i-k_i+1) +1$ 

	One can see that $n_i=n^i$.
	Now we consider the sequence $u_i = n_i-k_i+1$.
	$u_1 = n-k+1$ and $u_{i+1} = (n-k+1)u_i$.
	We deduce that $u_i = (n-k+1)^i$.
	Thus, by definition of $(u_i)$, $k_i = n_i-u_i+1 = n^i-(n-k+1)^i+1$.
\hfill $\Box$
\end{proof}

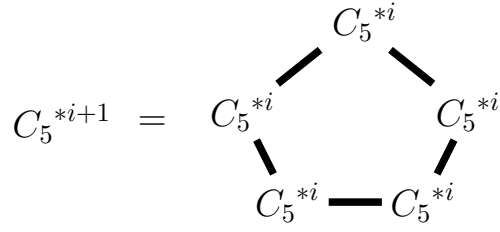
\begin{figure}[h!]
	\begin{center}
		\begin{tikzpicture}[scale=0.6]
			\node at (-4,1.75) {\Large ${C_5}^{*i+1}$};
			\node at (-2, 1.75) {\Large $=$};
			\node (v1) at (2.5,4){\Large ${\,\,\,C_5}^{*i}$};
			\node (v2) at (5,2){\Large ${C_5}^{*i}$};
			\node (v3) at (4,0){\Large ${C_5}^{*i}$};
			\node (v4) at (1,0){\Large ${C_5}^{*i}$};
			\node (v5) at (0,2){\Large ${C_5}^{*i}$};
			\draw[line width=1mm] (v1) -- (v2) -- (v3) -- (v4) -- (v5) -- (v1);
		\end{tikzpicture}
	\end{center}
\caption{Decomposition of the graph ${C_5}^{*i+1}$.}
\label{C5_i+1}
\end{figure}

\begin{theorem}
	For all $i \in \mathbb{N}^{*}$, the graph ${C_5}^{*i}$ realizes a $(( n, n-n^{ \frac{ log(3)}{ log(5)}} + 1 ))$ protocol (with $n = 5^i$).
\end{theorem}

\begin{proof}
	In Lemma \ref{G*i} with $n=5$ and $k=3$ (since the graph $C_5$ realizes a $((3,5))$ \cite{MS08}), we have $n=n_i = 5^i$ and $k_i = 5^i - 3^i + 1 = 5^i-5^{i \frac{log(3)}{log(5)}}+1=n-n^{\frac{log(3)}{log(5)}}+1$. \hfill $\Box$
\end{proof}

\section{Lower Bound}\label{sec:lowerbound}
\label{lb}
By the no-cloning theorem, it is not possible to get two separated copies of the secret starting from only one copy. Thus, if we consider a quantum secret sharing protocol with parameters $((k,n))$ we must have $k > \frac{n}{2}$.

We derive here less trivial lower bounds for our family of protocols. 
\begin{definition}
Let $B$ be a subset of vertices of a graph.
\begin{itemize}
\item $D \subseteq B$  is said to be "odd-wise in $B$" or "$B$-odd-wise" iff $D \cup Odd(D) \subseteq B$
\item $C \subseteq B$  is said to be "even-wise in $B$" or "$B$-even-wise" iff $C \cup Even(C) \subseteq B$
\end{itemize}
\end{definition}

\begin{lemma}
\label{smallCD}
Let $G=(V,E)$ be a graph which realizes a $((k,n))$ protocol.
Then, for all set $B$ of size $k$, there exists a set $X \subseteq B$ such that:
$
|X| \leq \frac{2}{3} \left( n-k+1 \right)
$
where $X$ is either a $B$-odd-wise set of size 1 mod 2 or a $B$-even-wise set.
\end{lemma}

\begin{proof}
First, let $\Gamma_B \in \mathcal{M}_{k,n-k}(\mathbb{F}_2)$ be a cut matrix of $G$ corresponding to the cut $(B,V\setminus B)$. 
We can see $\Gamma_B$ as the linear map that maps a set $D \subseteq B$ to its odd neighbourhood in $V \setminus B$: 
Consequently, any $B$-odd-wise set $D$ corresponds to a linear combination of the columns of the matrix $\Gamma_B$ which equals the null vector.
Therefore, $\{ D \big| D\ is\ odd\text{-}wise\ in\ B \} = Ker( \Gamma_B)$, and $t=dim( Ker( \Gamma_B)= k - dim( Im( \Gamma_B)) \geq 2k-n$.

As $|X\Delta Y|=|X|+|Y| \bmod 2$, the sets $\mathcal{D}_1=\{D \subseteq B, |D|=1 \bmod 2 \, D \,odd\text{-}wise \,in \,B\}$ and $\mathcal{C}_1=\{C \subseteq B,C \,even\text{-}wise\, in \,B\}$ are two affine subspaces having the same vector subspace $\mathcal{D}_0=\{D \subseteq B, |D|=0 \bmod 2 \wedge   D \,odd\text{-}wise \,in\, B\}$.

The dimension of $\mathcal{D}_0$ is $t-1$, therefore, by gaussian elimination its exists a set $X_0 \subseteq B$, $|X_0|=t-1$ such that it exists sets $C_1\in \mathcal{C}_1$ and $D_1\in \mathcal{D}_1$ satisfying $X_0\cap C_1=X_0 \cap D_1=\emptyset$. Thus $|C_1\cup D_1|\le k-t+1 \le n-k+1$.

Therefore $2|D_1\cup C_1|= |D_1|+|C_1|+|D_1\Delta C_1| \le 2(n-k+1)$ which implies that one of the three sets has cardinality smaller than $2(n-k+1)$, as the first set is odd-wise with odd cardinality and the two others are even-wise in $B$, which concludes the proof.
\hfill $\Box$
\end{proof}

By counting the even-wise and odd-wise sets and their possible completion into q-accessing sets, we get the following lower bound.
\begin{theorem}
There exists no graph $G$  that has an $((k,n))$ qQSS protocol with $k< \frac{n}{2} + \frac{n}{157}$.
\end{theorem}

\begin{proof}

We consider a graph $G=(V,E)$ which realizes a $((k,n))$ secret sharing protocol.

Any set of size $n-k$ is not c-accessing, therefore, any set $D$, with $|D| = 1 \bmod 2$ satisfies $ |D \cup Odd(D)| \geq n-k+1$. Consequently,  given a set $D$, with $|D|=1 \bmod 2$, it exists at most $ {n -(n-k+1) \choose {k-(n-k+1)}}={k-1 \choose 2k-n-1}$ sets  $B$ of size $k$ containing $D\cup Odd(D)$ and such that $D$ is odd-wise in $B$. 

Similarly,  any set of size $k$ is c-accessing, for any set $C$,  $ |C \cup Even(C)| \geq n-k+1$.
Therefore, given a set $C$  the number of sets $B$ of size $k$ containing $C$ and such that $C$ is even-wise in $B$ is at most ${k-1 \choose 2k-n-1}$. 

With Lemma \ref{smallCD}, each set $B \subseteq V$ of size $k$ contains either a $B$-odd-wise subset $D$ of size odd or a $B$-even-wise subset $C$ such that $|D| \leq \frac{2}{3}(n-k+1)$ or $|C| \leq \frac{2}{3}(n-k+1)$. Thus by counting twice all the sets of cardinality smaller then $\frac{2}{3}(n-k+1)$ (as a potential odd-wise or even-wise set) we can upper bound the set of possible cuts of size $k$ with $ {n \choose k} \leq 2\sum_{i=1}^{\frac{2}{3}(n-k+1)} {n \choose i} {k-1 \choose 2k-n-1}$.
The previous inequality implies that
$k > \frac{n}{2} + \frac{n}{157}$
when $n \to \infty$. \hfill $\Box$
\end{proof}
\begin{corollary}
There exists no $qQSS$ protocol for $n \ge 79$
\end{corollary}
\begin{proof}
By Gottesman's characterisation \cite{G00} a $qQSS$ protocol has a threshold $((k,2k-1))$. Moreover,
$k \ge n/2 +n/157$ using the previous lower bound. Therefore $k\le 159/4$ and the number of players $n=2k-1 \ge 79$.
\hfill $\Box$
\end{proof}

\end{document}